\documentclass[11pt]{article}

\usepackage{fullpage}

\newtheorem{theorem}{Theorem}%[section]
\newtheorem{lemma}[theorem]{Lemma}

\def\myendproof{\hfill{\vbox{\hrule\hbox{%
   \vrule height1.3ex\hskip0.8ex\vrule}\hrule}}}

\newcommand\real{{\rm I\kern-0.2em\rm R}}
\newcommand\itreal{{\it I\kern-0.35em\it R}}

\newenvironment{proof}{

\noindent{\bf Proof:}\ }{
\hfill \myendproof

}

\newenvironment{opt}[2]{
\samepage

\begin{eqnarray*}
& & \mbox{#1\ } #2 \\
& & \mbox{subject to\ }\left\{ % } MATCHING
\begin{array}{rcl@{\hspace{0.2in}}l}}{
\end{array}\right.
\end{eqnarray*}
}

\newcommand{\mathfn}[1]{\mathop{\rm #1}\nolimits }
\newcommand{\maximize}{\mathfn{maximize}}
\newcommand{\minimize}{\mathfn{minimize}}

\title{A Bound on the Sum 
  \\ of Weighted Pairwise Distances
  \\ of Points Constrained to Balls
  \thanks{
    Technical Report \#1103,
    School of Operations Research and Industrial Engineering,
    College of Engineering, Cornell University, Ithaca NY.
    This research was partially supported by \'Eva Tardos' NSF PYI grant.}
}

\date{September 1994}

\author{
  Neal E. Young
}

\begin{document}

\maketitle

\begin{abstract}
  We consider the problem of choosing Euclidean points
  to maximize the sum of their weighted pairwise distances,
  when each point is constrained to a ball centered at the origin.
  We derive a dual minimization problem and show strong duality holds
  (i.e., the resulting upper bound is tight)
  when some locally optimal configuration of points is affinely independent.
  We sketch a polynomial time algorithm 
  for finding a near-optimal set of points.
\end{abstract}

\section{Introduction}
We consider the following maximization problem $P(n,w,\ell)$:
\begin{opt}{$\maximize_{\{p_i\}}$}{\sum_{1 \le i < j \le n} w(i,j)d(p_i, p_j)}
  p_i & \in & \real^{n-1}  & (i=1,..,n);
  \\ ||p_i|| & \le & \ell(i)     & (i=1,..,n).
\end{opt}
Here each $w(i,j) \ge 0$ and each $\ell(i) \ge 0$ is fixed, 
$d(p,q)$ denotes the Euclidean distance between points $p$ and $q$,
and $||p||$ denotes the Euclidean length (distance from the origin)
of point $p$.

We derive the following dual problem $D(n,w,\ell)$:
\begin{opt}{$\minimize_{\{x_i\}}$}{
    \sqrt{\sum_{1\le i < j \le n} \frac{w^2(i,j)}{x_i x_j}}
    \times \sqrt{\sum_{i=1}^n \ell^2(i) x_i}
    \times \sqrt{\sum_{i=1}^n x_i}}
  x_i & \in & \real & (i=1,..,n);
  \\ x_i & \ge & 0     & (i=1,..,n).
\end{opt}
Throughout the paper, $\frac{0}{0}$ is defined to be $0$.

We show that the value of the maximization problem
is at most the value of the minimization problem.
We use a physical interpretation of the two problems to show that
the values are equal provided the maximization problem
admits a set of points $\{p_i\}$ that is both affinely independent
and stationary
(i.e., the gradient of the objective function is a nonnegative combination
of the gradients of the active constraints,
a necessary condition at any local maximizer of $P(n,w,\ell)$).

We sketch how a near-optimal solution to the problem
can be found in polynomial time via the ellipsoid method.

\section{Related Work}

The case $w(i,j) = \ell(i) = 1$
(in which the optimal points are given by the vertices
of the regular $n$-simplex, achieving a value of $n\sqrt{{n \choose 2}}$)
was previously considered by \cite{Li75}.
Our Lemma \ref{weak duality}
generalizes a bound in that paper.

Specific instances of $P(n,w,\ell)$ 
were studied to obtain geometric inequalities 
that were used to analyze approximation algorithms
for finding low-degree, low-weight spanning trees
in Euclidean spaces \cite{KhullerRY94}.

Goemans and Williamson \cite{GoemansW94}
consider related problems with applications
to approximating the maximum cut in a graph
and to maximizing the number of satisfied clauses 
in a CNF formula.
We modify their approach to solving their problems
to obtain a polynomial time algorithm for ours.

\section{A Dual Problem}

\begin{lemma}
  For any $n$, $w$, and $\ell$,
  the value of the maximization problem $P(n,w,\ell)$
  is at most the value of the minimization problem $D(n,w,\ell)$.
  \label{weak duality}
\end{lemma}
\begin{proof}
  Fix any $n$, $w$, and $\ell$.
  Fix any set of points $\{p_i\}$ and values $\{x_i\}$
  meeting the constraints of $P(n,w,\ell)$ and $D(n,w,\ell)$, respectively.
  Let $A(i,j) = \frac{w(i,j)}{\sqrt{x_i x_j}}$
  and $B(i,j) = \sqrt{x_i x_j} d(p_i,p_j)$
  for $1 \le i < j \le n$.
  Then, by the Cauchy-Schwartz inequality $A \cdot B \le \|A\|\times\|B\|$
  (where $A$ and $B$ are interpreted as $n \choose 2$-dimensional vectors,
  and $\cdot$ denotes the dot product):
  \begin{equation}
    \sum_{i<j} w(i,j) d(p_i,p_j)
    \le \sqrt{\sum_{i<j} \frac{w^2(i,j)}{x_i x_j}}
    \times \sqrt{\sum_{i<j} x_i x_j d^2(p_i, p_j)}.
    \label{one}
  \end{equation}
  It remains only to show
  \begin{displaymath}
    \sum_{i<j} x_i x_j d^2(p_i, p_j)
    \le \left(\sum_i x_i\right)
    \times \left(\sum_i \ell^2(i) x_i\right).
  \end{displaymath}

  Expanding the left-hand side,
  \begin{eqnarray}
    \lefteqn{\sum_{i<j} x_i x_j d^2(p_i, p_j)}
    \nonumber
    \\ & = & \frac{1}{2}\sum_{i,j} x_i x_j (p_i-p_j)\cdot(p_i-p_j)
    \nonumber
    \\ & = & \frac{1}{2}\sum_{i,j} x_i x_j
    (p_i\cdot p_i - 2 p_i\cdot p_j + p_j\cdot p_j)
    \nonumber
    \\ & \le & \sum_{i,j} x_i x_j (\ell^2(i) - p_i\cdot p_j)
    \label{two}
    \\ & = & \left(\sum_i x_i\right)\times\left(\sum_i x_i \ell^2(i)\right)
    - \left(\sum_i x_i p_i\right)\cdot\left(\sum_i x_i p_i\right)
    \nonumber
    \\ & = & \left(\sum_i x_i\right)
    \times \left(\sum_i x_i \ell^2(i)\right)
    - \left|\left|\sum_i x_i p_i\right|\right|^2
    \nonumber
    \\ & \le &  \left(\sum_i x_i\right)
    \times \left(\sum_i x_i \ell^2(i)\right).
    \label{three}
  \end{eqnarray}
\end{proof}

\begin{lemma}
  Fix any $n$, $w$, and $\ell$.
  Suppose the maximization problem $P(n,w,\ell)$
  admits a set of points $\{p_i\}$
  that is both stationary and affinely independent.
  Then the values of the two problems are equal.
  Further, there exists $\{x_i\}$ such that
  \begin{equation}
    x_i p_i = \sum_j w(i,j) \frac{p_i - p_j}{d(p_i,p_j)}
    \label{xeqn}
  \end{equation}
  (where $x_i = 0$ in case $\|p_i\| < \ell_i$,
  and $w(i,j) = w(j,i)$ and $w(i,i) = 0$),
  and $\{p_i\}$ and $\{x_i\}$ are global optima for the two problems.
  \label{strong duality}
\end{lemma}
\begin{proof}
  Fix any $n$, $w$, and $\ell$.
  Consider the objective function $\Phi(\{p_i\}) = \sum_{ij} w(i,j) d(p_i,p_j)$
  of $P(n,w,\ell)$.  That $\{p_i\}$ is stationary
  means that the gradient of $\Phi$ is a nonnegative combination
  of the gradients of the constraints of $P(n,w,\ell)$ active at $\{p_i\}$.
  By elementary calculus, 
  the gradient of $\Phi$ consists of a vector $f_i$ for each point $p_i$,
  with each $f_i$ equal to the right-hand side of (\ref{xeqn}).
  The only constraint on $p_i$ is $\|p_i\| \le \ell(i)$,
  whose gradient (again by elementary calculus) 
  is a nonnegative multiple of $p_i$.
  Thus, for each $i$, there exists an $x_i \ge 0$ such that (\ref{xeqn}) holds.
  Note that if $\|p_i\| < \ell(i)$, then the constraint is not active,
  so that $f_i$ must be the zero vector.
  In this case we take $x_i = 0$.

  We will show that each inequality in Lemma \ref{weak duality}
  is tight for these $\{p_i\}$ and $\{x_i\}$.
  Inequality (\ref{three}) is tight because, by (\ref{xeqn}),
  $\sum_i x_i p_i$ is the zero vector.
  Inequality (\ref{two}) is tight
  because $\|p_i\| < \ell(i)$ only if $x_i = 0$.
  
  Inequality (\ref{one}) is tight
  provided the vector $A$ (in the proof of Lemma \ref{weak duality})
  is a scalar multiple of $B$.
  Assume $\{p_i\}$ is affinely independent.
  Then, considering $\{x_i\}$ and $\{p_i\}$ fixed
  and $\{w(i,j)\}$ as the set of unknowns (i.e., reversing their roles),
  (\ref{xeqn}) uniquely determines each $w(i,j)$.
  Since
  \begin{equation}
    w(i,j) = \frac{x_i x_j d(p_i, p_j)}{\sum_k x_k} ~~ (1\le i < j \le n)
    \label{weqn}
  \end{equation}
  is consistent with (\ref{xeqn}) 
  (check this by substitution for $w(i,j)$ in (\ref{xeqn})),
  it follows that (\ref{weqn}) necessarily holds.
  Thus, $A$ is a scalar multiple of $B$ and Inequality (\ref{one}) is tight.
\end{proof}

A physical model for the quantities involved is as follows.
Consider a physical system of $n$ points $\{p_i\}$.
Each point $p_i$ is constrained to a ball of radius $\ell(i)$ 
centered at the origin.
For each pair of points $(p_i,p_j)$,
$p_i$ repels $p_j$ (and vice versa) with a force of magnitude $w(i,j)$.

Under this interpretation, 
each vector $f_i$ in the proof corresponds to the force on $p_i$,
and $x_i$ is the magnitude of this force, divided by $\|p_i\|$.

\section{Solving $P(n,w,\ell)$ in Polynomial Time}

If the instance of $P(n,w,\ell)$ is small or has a high degree of symmetry,
the dual problem $D(n,w,\ell)$ might yield a function
that can be minimized directly by symbolic methods.
In general, 
it is possible to solve $P(n,w,\ell)$
(to any given degree of precision)
in polynomial time using semi-definite programming,
following the approach in \cite{GoemansW94}.

Those authors consider a related problem $GW(w,n)$:
\begin{opt}{$\maximize_{\{p_i\}}$}{
    \sum_{1 \le i < j \le n} w(i,j)d^2(p_i, p_j)}
  p_i & \in & \real^{n}  & (i=1,..,n);
  \\ \|p_i\| & = & 1     & (i=1,..,n).
\end{opt}
The authors show how to solve this problem in polynomial time
by formulating it as a semi-definite program,
and how to round a (near-)optimal set of points $\{p_i\}$
to obtain an approximate solution to a corresponding max-cut problem.
This approach yielded the first polynomial-time approximation algorithm
achieving a performance guarantee better than two for the max-cut problem.

We briefly sketch their aproach for solving $GW(w,n)$
and how it can be modified to solve $P(w,n,\ell)$.
The connection between sets of points 
and positive semi-definite matrices is the following:
an $n\times n$ symmetric matrix $Y$ is positive semi-definite
if and only if there exists a set of $n$ points $\{p_i\}$ in $\real^n$
such that $Y_{ij} = p_i \cdot p_j$.
Thus, $GW(w,n)$ is equivalent to following:
\begin{opt}{$\maximize_{\{Y\}}$}{
    \sum_{ij} w(i,j)(2-2 Y_{ij})}
  Y & \mbox{\small is} & 
        \mbox{\small an $n\times n$ symmetric, positive semi-definite matrix};
  \\ Y_{ii} & = & 1     & (i=1,..,n).
\end{opt}

The space of $n\times n$ symmetric, positive semi-definite matrices
admits a polynomial time separation oracle
because a symmetric matrix $Y$ is positive semi-definite
if and only if $x^T Y x \ge 0$ for each $x\in\real^n$,
and in fact it suffices to check each eigenvector $x$ of $Y$.
Thus, combining the constraint that $Y$ is positive semi-definite
with arbitrary linear inequalities on the elements of $Y$
yields a convex space with a polynomial time separation oracle.
Approximate feasibility of such a problem is testable in polynomial time
via the ellipsoid method.  
Thus, $GW(n)$ can be solved to near-optimality in polynomial time.

A similar approach can be used to solve $P(n,w,\ell)$ in polynomial time.
In particular, $P(n,w,\ell)$ 
corresponds to the following semi-definite program:
\begin{opt}{$\maximize_{\{Y\}}$}{
    \sum_{ij} w(i,j)\sqrt{Y_{ii}+Y_{jj}-2 Y_{ij}}}
  Y & \mbox{is} & 
        \mbox{an $n\times n$ symmetric, positive semi-definite matrix};
  \\ Y_{ii} & \le & \ell(i) & (i=1,..,n).
\end{opt}
Since $\sum_{ij} w(i,j)\sqrt{Y_{ii}+Y_{jj}-2 Y_{ij}}$
is a concave function in $\{Y_{ij}\}$
whose gradient can be computed in polynomial time,
the above program also admits a separation oracle
sufficient to solve it to near-optimality in polynomial time
using the ellipsoid method.

\section{Open Problems}
It would be interesting to obtain a more efficient algorithm
for solving $P(w,n,\ell)$ than is obtained by reducing to the ellipsoid method.
Especially interesting would be a primal-dual algorithm
along the lines of traditional ``combinatorial'' algorithms
for solving or approximating linear programs.
It is not clear how to achieve such algorithms 
in the semi-definite setting.

Similarly, the only known method for achieving a better factor than two
for the max-cut problem is by reduction to semi-definite programming.
Goemans and Williamson leave open the problem 
of finding a more efficient algorithm that beats a factor of two.
A more efficient algorithm for $P(n,w,\ell)$ (with each $\ell(i) = 1$)
would solve this,
because applying their randomized rounding technique
to $P(n,w,\ell)$ also yields an approximation algorithm for max-cut
with performance guarantee better than two.

On the other hand, consider the generalization of $GW(n,w)$
in which the objective function is replaced by
$\sum_{ij} w(i,j) d^{2+\epsilon}(p_i,p_j)$ for some $\epsilon \ge 0$.
For $\epsilon > 0$, 
applying Goemans and Williamson's approach 
to this program rather than $GW(n,w)$
would provide a better approximation to max-cut.
Is the generalization solvable in polynomial time for some $\epsilon > 0$?

\section{Acknowledgements}
Thanks to Michael Todd for helpful discussions 
and comments on an earlier draft of this paper.

\bibliographystyle{plain}
\bibliography{full,tech-report}

\begin{thebibliography}{1}

\bibitem{GoemansW94}
Michel Goemans and David Williamson.
\newblock .878-approximation algorithms for {MAX CUT} and {MAX 2SAT}.
\newblock In {\em Proc. of the 26th Ann. ACM Symp. on Theory of Computing},
  1994.

\bibitem{KhullerRY94}
Samir Khuller, Balaji Raghavachari, and Neal Young.
\newblock Low-degree spanning trees of small weight.
\newblock In {\em Proc. of the 26th Ann. ACM Symp. on Theory of Computing},
  1994.

\bibitem{Li75}
J.~N. Lillington.
\newblock Some extremal properties of convex sets.
\newblock {\em Math.~Proc.~Cambridge Philosophical Society}, 77:515--524, 1975.

\end{thebibliography}

\end{document}